\documentclass[a4paper, 11pt, twoside, notitlepage]{amsart}

\usepackage{amsmath,amscd}
\DeclareMathOperator*{\argmin}{arg\,min}
\usepackage{amssymb}
\usepackage{amsthm}
\usepackage{comment}
\usepackage{graphicx, xcolor}
\usepackage[margin=1in]{geometry}
\usepackage{xcolor}
\usepackage{mathrsfs}
\usepackage[ocgcolorlinks, linkcolor=blue]{hyperref}

\usepackage{bm}
\usepackage{algorithm}
\usepackage{algorithmic}
\usepackage{url}

\usepackage[utf8]{inputenc}
\usepackage{mathtools}
\usepackage{esint}
\usepackage{tikz}
\usetikzlibrary{arrows.meta, positioning, calc}
\usepackage{relsize}
\usepackage[shortlabels]{enumitem}
\numberwithin{equation}{section}

\allowdisplaybreaks

\mathtoolsset{showonlyrefs}

\newtheorem{theorem}{Theorem}[section]
\newtheorem{lemma}[theorem]{Lemma}
\newtheorem{definition}[theorem]{Definition}
\newtheorem{corollary}[theorem]{Corollary}

\newtheorem{proposition}[theorem]{Proposition}

\newtheorem{remark}[theorem]{Remark}
\newtheorem{assumption}[theorem]{Assumption}

\title[MFG for Network SEIR Dynamics]{Learning Contact Policies for SEIR Epidemics on Networks: A Mean-Field Game Approach}

\author[W. Wang]{Weinan Wang}
\address{Department of Mathematics, University of Oklahoma, Norman, OK, USA}
\email{ww@ou.edu}

\date{\today}

\begin{document}
	
	\maketitle

\begin{abstract}
    In this paper, we develop a mean-field game model for SEIR epidemics on heterogeneous contact networks, where individuals choose state-dependent contact effort to balance infection losses against the social and economic costs of isolation. The Nash equilibrium is characterized by a coupled Hamilton--Jacobi--Bellman/Kolmogorov system across degree classes. An important feature of the SEIR setting is the exposed compartment: the incubation period separates infection from infectiousness and changes incentives after infection occurs. In the baseline formulation, exposed agents optimally maintain full contact, while susceptible agents reduce contact according to an explicit best-response rule driven by infection pressure and the value gap. We also discuss extensions that yield nontrivial exposed precaution by introducing responsibility or compliance incentives. We establish existence of equilibrium via a fixed-point argument and prove the uniqueness under a suitable monotonicity condition. The analysis identifies a delay in the onset of precaution under longer incubation, which can lead to weaker behavioral responses and larger outbreaks. Numerical experiments illustrate how network degree and the cost exponent shape equilibrium policies and epidemic outcomes.

\medskip

\end{abstract}
    
	
\section{Introduction}
\label{sec:introduction}

Classical compartmental epidemic models, dating back to the seminal work of Kermack and McKendrick, treat contact rates as exogenous parameters and focus on the resulting population-level dynamics \cite{Kermack1927,AndersonMay1991}. In practice, contact behavior responds to perceived risk, information, and policies, producing feedback between disease prevalence and transmission opportunities \cite{Funk2010,Verelst2016,Fenichel2011}. This feedback can substantially alter epidemic trajectories and can also bias inference and forecasting when it is ignored \cite{Fenichel2011,Eksin2017}. A complementary perspective models individual choices explicitly, for instance through differential-game or game-theoretic formulations of distancing or protection \cite{Reluga2010,Bauch2004}. Network structure provides an additional layer of heterogeneity that is crucial for both transmission and behavior. A large literature shows how degree distributions, correlations, and mixing patterns reshape epidemic thresholds, early growth, and final size \cite{KeelingEames2005,Newman2002Spread,PastorSatorras2015,Volz2008,Boguna2002,Kiss2017}. In particular, degree-structured and correlated networks are naturally encoded through $P(k)$ and $P(k'|k)$, which allow one to represent heterogeneous connectivity and assortativity in a tractable way \cite{Newman2002Spread,Boguna2002}.

Mean-field game (MFG) theory offers a principled framework for modeling large populations of interacting decision makers, where each agent optimizes an individual objective against an aggregate state \cite{Lasry2007,Huang2006,CarmonaDelarue2018,GomesSaude2014}. From a computational standpoint, forward--backward schemes and related finite-difference/iterative methods are standard tools for MFG systems \cite{Achdou2010}. Moreover, the interpretation of equilibrium computation as a learning procedure has become increasingly relevant for policy and behavioral applications \cite{Cardaliaguet2017,CarmonaLauri2019}. In epidemic modeling, several works have used MFG or mean-field equilibrium ideas to study contact reduction and the inefficiency gap between individual and socially optimal responses \cite{Elie2020,Cho2020,Bremaud2022,Bremaud2024NPI,Bremaud2025}. 
A complementary control-theoretic direction steers stochastic compartmental epidemic models by PDE-constrained control of the associated Fokker--Planck equation; cf \cite{ParkinsonRoy2026FP}.
Existing network-based MFG epidemic models have largely focused on SIR-type dynamics \cite{Bremaud2025}. For pathogens with a meaningful latent period---with SARS-CoV-2 as a motivating example---the exposed compartment matters: infection and infectiousness are separated in time, and pre-symptomatic transmission can be significant \cite{He2020}. This separation changes incentives after infection occurs, and it interacts with network heterogeneity in ways that are not present in SIR-based models.

In this paper, we develop a mean-field game framework for SEIR dynamics on heterogeneous networks. We derive the coupled HJB--Kolmogorov system across degree classes, using the compensator identity to represent a one-time infection loss as an equivalent running cost \cite{Bremaud1981}. We show that in the baseline formulation exposed agents optimally maintain full contact, and we discuss extensions that yield nontrivial exposed precaution through responsibility or compliance incentives. We prove existence of equilibrium via a Schauder fixed-point argument and provide uniqueness under a suitable monotonicity condition. The analysis identifies an incubation-induced delay in the onset of precaution when the latent period is longer, which can lead to weaker equilibrium contact reductions and larger outbreaks. Numerical experiments illustrate how network degree and the isolation-cost exponent shape equilibrium policies and epidemic outcomes.

\subsection{Related Work}

Behavioral responses in epidemics have been widely modeled through risk-dependent transmission and information-driven behavior, with surveys and systematic reviews summarizing many mechanisms and modeling choices \cite{Funk2010,Verelst2016}. Game-theoretic treatments of protective behavior include vaccination and distancing models that expose free-riding and externalities \cite{Bauch2004,Reluga2010}. On the network side, a broad theory connects epidemic dynamics to degree heterogeneity and mixing structure; see \cite{KeelingEames2005,PastorSatorras2015,Kiss2017} and the degree-based random network formulation in \cite{Newman2002Spread,Volz2008}. Correlations encoded by $P(k'|k)$ (Markovian networks) are known to alter thresholds and equilibrium conditions \cite{Boguna2002}.

For MFG theory and computation, we refer to the foundational and modern references \cite{Lasry2007,Huang2006,CarmonaDelarue2018,GomesSaude2014,Achdou2010}. Learning-based viewpoints for equilibrium computation include fictitious play and mean-field reinforcement learning \cite{Cardaliaguet2017,CarmonaLauri2019}. In epidemic applications, mean-field equilibrium approaches to contact reduction and policy comparison are developed in \cite{Elie2020,Cho2020} and in social-structure models in \cite{Bremaud2022,Bremaud2024NPI,Bremaud2025}. The present paper differs by focusing on SEIR dynamics on heterogeneous networks and by isolating how the exposed compartment reshapes equilibrium incentives and timing.
\section{The SEIR Network Model}
\label{sec:model}
	
We consider a population of $N$ individuals ($N$ large) situated on a network where nodes represent individuals and edges represent potential transmission contacts. Each individual has degree $k$ (number of neighbors) drawn from distribution $P(k)$. Following the approach of Bremaud et al. \cite{Bremaud2025}, we consider Markovian networks characterized by $P(k)$ and the degree correlation matrix $G_{kk'} = P(k'|k)$.
	
\subsection{Compartmental Dynamics}
	
Individuals can be in one of four states: Susceptible ($S$), Exposed ($E$), Infectious ($I$), or Recovered ($R$). Let $S_k(t)$, $E_k(t)$, $I_k(t)$, and $R_k(t)$ denote the fractions of individuals with degree $k$ in each state at time $t$, with normalization $S_k(t) + E_k(t) + I_k(t) + R_k(t) = 1$.
The disease progression follows standard SEIR dynamics with network-dependent transmission:
\begin{align}
\dot{S}_k(t) &= - \lambda_k(t) S_k(t) \label{eq:S_dot} \\
\dot{E}_k(t) &= \lambda_k(t) S_k(t) - \sigma E_k(t) \label{eq:E_dot} \\
\dot{I}_k(t) &= \sigma E_k(t) - \gamma I_k(t) \label{eq:I_dot} \\
\dot{R}_k(t) &= \gamma I_k(t) \label{eq:R_dot}
\end{align}
where $\sigma^{-1}$ is the mean incubation period, $\gamma$ is the recovery rate.
	
\subsection{Network-Dependent Force of Infection}

We incorporate network structure through the degree correlation matrix $G_{kk'}=P(k'|k)$ (Markovian networks).
Let $\Theta_k(t)$ denote the probability that a randomly chosen neighbor of a degree-$k$ node is effectively infectious. In the baseline SEIR setting, only infectious individuals transmit. Infectious individuals do not optimize their contact rate (they have no further incentive to reduce contacts and face no effort cost); we set $n_k^I \equiv 1$ for all $k$. The effective infection pressure from a randomly chosen neighbor of a degree-$k$ node is therefore
\begin{equation}
\Theta_k(t)\;=\;\sum_{k'} G_{kk'}\, I_{k'}(t).
\label{eq:Theta_k}
\end{equation}
A susceptible individual of degree $k$ has $k$ neighbors, each contacted at effort $n_k^S(t)$, hence the force of infection is
\begin{equation}
\lambda_k(t)\;=\;\beta\, k\, n_k^{S}(t)\,\Theta_k(t).
\label{eq:lambda_k}
\end{equation}

\begin{remark}[Why $n^I \equiv 1$]
Once infectious, an individual has no remaining infection-risk incentive and faces no effort-cost term in \eqref{eq:HJB_I}. Introducing $n^I$ as a free variable without a corresponding cost would be internally inconsistent, so we set $n^I \equiv 1$.
\end{remark}

If exposed individuals transmit at reduced rate $\beta_E=\kappa\beta$ with $\kappa\in[0,1]$, then we instead set
\begin{equation}
\Theta_k(t)\;=\;\sum_{k'} G_{kk'}\Big(I_{k'}(t)+\kappa\, n_{k'}^{E}(t) E_{k'}(t)\Big),
\label{eq:Theta_k_presym}
\end{equation}
and keep \eqref{eq:lambda_k} unchanged.

If the network is uncorrelated, $G_{kk'}=\frac{k'P(k')}{\langle k\rangle}$, so $\Theta_k(t)$ becomes independent of $k$:
\begin{equation}
\Theta(t)\;=\;\frac{1}{\langle k\rangle}\sum_{k'} k'P(k')\, I_{k'}(t),
\qquad
\lambda_k(t)=\beta\,k\,n_k^S(t)\,\Theta(t).
\label{eq:Theta_uncorr}
\end{equation}
	\subsection{Some Assumptions}
	
	\begin{assumption}[Control in the exposed compartment]
We allow exposed individuals to choose a contact level $n_k^E(t)\in[\mathfrak n_{\min},1]$ to model awareness, compliance, or quarantine decisions.
Under the baseline cost structure and purely symptomatic transmission, however, the exposed-state control does not affect the agent's own transition intensity, and the Nash best response is $n_k^{E*}(t)=1$ (see \eqref{eq:opt_E} and Remark~\ref{rem:nE_nash}).
Nontrivial exposed precautions ($n_k^{E*}<1$) arise only under extensions that internalize pre-symptomatic transmission externalities or include responsibility/compliance penalties.
\end{assumption}
	
	\begin{assumption}[Asymptomatic Transmission]
		We consider two scenarios: (1) Exposed individuals are not infectious, with transmission occurring only from $I$ individuals; (2) Exposed individuals can transmit at reduced rate $\beta_E = \kappa \beta$ with $0 \leq \kappa \leq 1$, reflecting pre-symptomatic transmission.
	\end{assumption}

\section{Mean-Field Game Formulation}
\label{sec:mfg_formulation}
	
\subsection{Individual Cost Structure}
	
Each individual of degree $k$ aims to minimize an expected cost over time horizon $[0, T]$. The control variables are $(n^S(t), n^E(t))$, the contact-reduction efforts in states $S$ and $E$ respectively; in states $I$ and $R$ there is no active control ($n^I = n^R \equiv 1$). The individual cost is
\begin{equation}
J_k = \mathbb{E} \left[ \int_0^T \left(
    r_I\,\lambda_k(t,n^S)\,\mathbb{I}_{\{x(t)=S\}}
    + f_k(n(t))\,\mathbb{I}_{\{x(t)\in\{S,E\}\}}
    + C_E\,\mathbb{I}_{\{x(t)=E\}}
    + C_I\,\mathbb{I}_{\{x(t)=I\}}
\right) dt + \Psi(x(T)) \right], \label{eq:individual_cost}
\end{equation}
where $n(t) = n^S(t)\,\mathbb{I}_{\{x(t)=S\}} + n^E(t)\,\mathbb{I}_{\{x(t)=E\}}$ is the state-dependent control, $f_k(\cdot)$ is the social/economic isolation cost, $C_E,C_I$ are health costs, and $\Psi$ is a terminal cost. Following Bremaud et al., we take
\begin{equation}
f_k(n) = k^\epsilon \left( \frac{1}{n} - 1 \right), \quad \epsilon \in \{0,1\}. \label{eq:social_cost}
\end{equation}

\begin{remark}[Infection lump-sum cost via the compensator identity]
\label{rem:compensator}
The term $r_I$ represents a one-time cost paid at the (random) moment of infection $\tau = \inf\{t : x(t)=E\}$. By the Doob--Meyer compensator identity for a counting process with stochastic intensity $\lambda_k(t,n^S)\,\mathbb{I}_{\{x(t)=S\}}$,
\[
\mathbb{E}\!\left[r_I\,\mathbb{I}_{\{\tau \le T\}} \;\middle|\; x(0)=S\right]
= \mathbb{E}\!\left[\int_0^T r_I\,\lambda_k(t,n^S)\,\mathbb{I}_{\{x(t)=S\}}\,dt\;\middle|\; x(0)=S\right].
\]
Writing it as a running cost inside the integral (rather than a terminal/stopping-time payment) is therefore exact and ensures the HJB equation \eqref{eq:HJB_S} is the correct necessary condition for the stated optimization problem. In particular the jump term $r_I + U_k^E - U_k^S$ in \eqref{eq:HJB_S} arises directly from differentiating this running cost.
\end{remark}
	
\subsection{Value Functions and HJB Equations}
	
Define the value function $U_k^x(t)$ as the minimum expected remaining cost for an individual of degree $k$ in state $x$ at time $t$:
\begin{equation}
\begin{split}
U_k^x(t) = \min_{(n^S,n^E)} \mathbb{E} \Bigl[
    \int_t^T \bigl(
        &r_I\lambda_k(s,n^S)\mathbb{I}_{\{x(s)=S\}}
        + f_k(n(s))\mathbb{I}_{\{x(s)\in\{S,E\}\}} \\
        &+ C_E\mathbb{I}_{\{x(s)=E\}}
        + C_I\mathbb{I}_{\{x(s)=I\}}
    \bigr)\,ds
    + \Psi(x(T)) \;\Big|\; x(t)=x \Bigr].
\end{split}
\label{eq:value_function}
\end{equation}
Applying dynamic programming principles leads to the Hamilton-Jacobi-Bellman equations. For susceptible individuals:
\begin{equation}
-\frac{dU_k^S}{dt} = \min_{n^S \in [\mathfrak{n}_{\text{min}}, 1]} \left[ \lambda_k(t, n^S) \left( r_I + U_k^E(t) - U_k^S(t) \right) + f_k(n^S) \right] \label{eq:HJB_S}
\end{equation}
where $r_I$ is the infection cost incurred at the moment of infection, and $\lambda_k(t, n^S) = \beta k n^S \Theta_k(t)$.
For exposed individuals:
\begin{equation}
-\frac{dU_k^E}{dt} = \min_{n^E \in [\mathfrak{n}_{\text{min}}, 1]} \left[ \sigma \left( U_k^I(t) - U_k^E(t) \right) + C_E + f_k(n^E) \right] \label{eq:HJB_E}
\end{equation}
For infectious and recovered individuals:
\begin{align}
-\frac{dU_k^I}{dt} &= \gamma \left( U_k^R(t) - U_k^I(t) \right) + C_I \label{eq:HJB_I} \\
-\frac{dU_k^R}{dt} &= 0 \label{eq:HJB_R}
\end{align}
with terminal conditions $U_k^x(T) = \Psi(x)$.
	
\subsection{Rigorous Derivation of HJB Equations}
	
	\begin{theorem}[HJB Derivation for SEIR-MFG]
		The value functions $U_k^x(t)$ defined in (\ref{eq:value_function}) satisfy the HJB equations (\ref{eq:HJB_S})-(\ref{eq:HJB_R}) under appropriate regularity conditions.
	\end{theorem}
	
	\begin{proof}
		We derive the HJB equation for susceptible individuals in detail; the other cases follow similarly. Consider a susceptible individual of degree $k$ at time $t$. By the dynamic programming principle for a small time increment $h > 0$, and using the compensator identity for $r_I$ (Remark~\ref{rem:compensator}):
		\begin{equation*}
\begin{aligned}
U_k^S(t) = \min_{n^S(\cdot)} \mathbb{E} \bigg[ \int_t^{t+h} \Big( &r_I \lambda_k(s, n^S) \mathbb{I}_{\{x(s)=S\}} + f_k(n^S(s)) \mathbb{I}_{\{x(s) \in \{S, E\}\}} \\
&+ C_E \mathbb{I}_{\{x(s)=E\}} + C_I \mathbb{I}_{\{x(s)=I\}} \Big) \, ds \\
&+ U_k^{x(t+h)}(t+h) \;\bigg|\; x(t)=S \bigg].
\end{aligned}
\end{equation*}
		For sufficiently small $h$, the probability of multiple state transitions is $o(h)$. The individual remains susceptible with probability $1 - \lambda_k(t, n^S(t))h + o(h)$ and becomes exposed with probability $\lambda_k(t, n^S(t))h + o(h)$. Using the compensator representation of $r_I$, the dynamic programming principle gives:
		\begin{align*}
		U_k^S(t) &= \min_{n^S} \left\{ \bigl(r_I\lambda_k(t,n^S) + f_k(n^S)\bigr)h + \left(1 - \lambda_k(t, n^S)h\right) U_k^S(t+h) \right. \\
		&\quad \left. + \lambda_k(t, n^S)h\, U_k^E(t+h) + o(h) \right\}
		\end{align*}
		Rearranging and taking $h \to 0$:
		\begin{align*}
		\frac{U_k^S(t+h) - U_k^S(t)}{h} &= \min_{n^S} \left\{ r_I\lambda_k(t,n^S) + f_k(n^S) + \lambda_k(t, n^S) \left( U_k^E(t+h) - U_k^S(t+h) \right) \right\} + o(1)
		\end{align*}
		Taking the limit $h \to 0$:
		\begin{equation*}
		-\frac{dU_k^S}{dt} = \min_{n^S} \left\{ \lambda_k(t,n^S)\bigl(r_I + U_k^E(t) - U_k^S(t)\bigr) + f_k(n^S) \right\}
		\end{equation*}
		which is exactly equation (\ref{eq:HJB_S}).
		For exposed individuals, the derivation is similar but accounts for the Markovian transition $E\to I$ with intensity $\sigma$: $\sigma$:
		\begin{align*}
		U_k^E(t) &= \min_{n^E} \left\{ (C_E + f_k(n^E))h + (1 - \sigma h) U_k^E(t+h) + \sigma h U_k^I(t+h) + o(h) \right\}
		\end{align*}
		leading to:
		\begin{equation*}
		-\frac{dU_k^E}{dt} = \min_{n^E} \left\{ C_E + f_k(n^E) + \sigma \left( U_k^I(t) - U_k^E(t) \right) \right\}
		\end{equation*}
		which is equation (\ref{eq:HJB_E}). The derivations for infectious and recovered states are straightforward since no optimization is involved in those states (once infectious, the disease progression is autonomous).
	\end{proof}
	
\subsection{Optimal Effort Policies}

\begin{definition}[Degree-scaled infection pressure]
\label{def:thetabar}
Define $\bar\theta_k(t) := k\,\Theta_k(t)$, so that the force of infection takes the compact form
\[
\lambda_k(t,n^S) = \beta\,n^S\,\bar\theta_k(t).
\]
In the uncorrelated case $\bar\theta_k(t) = k\,\Theta(t)$ with $\Theta(t)$ given by \eqref{eq:Theta_uncorr}.
\end{definition}
The optimal effort levels are obtained from the minimization in the HJB equations. For susceptible individuals:
	\begin{equation}
	n_k^{S*}(t) = \arg\min_{n^S \in [\mathfrak{n}_{\text{min}}, 1]} \left[ \beta n^S \bar\theta_k(t) \left( r_I + U_k^E(t) - U_k^S(t) \right) + f_k(n^S) \right] \label{eq:opt_S}
	\end{equation}
	
	\begin{theorem}[Closed-Form Optimal Control]
		For the social cost function $f_k(n) = k^\epsilon \left( \frac{1}{n} - 1 \right)$, the optimal effort for susceptible individuals is:
		\begin{equation}
n_k^{S*}(t)=
\begin{cases}
1, & \bar\theta_k(t)\,\Delta U_k(t)\le 0,\\[2mm]
\min\left\{ 1,\ \max\left\{ \mathfrak{n}_{\text{min}},\ \sqrt{ \dfrac{k^\epsilon}{\beta\, \bar\theta_k(t)\, \Delta U_k(t)} } \right\} \right\},
& \bar\theta_k(t)\,\Delta U_k(t)>0,
\end{cases}
\qquad \Delta U_k(t) := r_I + U_k^E(t) - U_k^S(t).
\label{eq:n_closed_form}
\end{equation}
		Since $\bar\theta_k(t)=k\,\Theta_k(t)$, this is equivalently
		\begin{equation}
		n_k^{S*}(t) = \min\left\{1,\,\max\left\{\mathfrak{n}_{\text{min}},\,
		\sqrt{\frac{k^{\epsilon-1}}{\beta\,\Theta_k(t)\,\Delta U_k(t)}}\right\}\right\}.
		\label{eq:n_closed_form_alt}
		\end{equation}
\noindent
(When $\Theta_k(t)\,\Delta U_k(t)\le 0$, the minimizer is $n_k^{S*}(t)=1$ by monotonicity of the Hamiltonian in $n$.)

	\end{theorem}
	
	\begin{proof}
		The minimization problem in (\ref{eq:opt_S}) with $\lambda_k(t,n^S)=\beta n^S \bar\theta_k$ for interior solutions requires solving:
		\begin{equation*}
		\frac{\partial}{\partial n^S} \left[ \beta n^S \bar\theta_k(t) \Delta U_k(t) + k^\epsilon \left( \frac{1}{n^S} - 1 \right) \right] = 0
		\end{equation*}
		This gives:
		\begin{equation*}
		\beta \bar\theta_k(t) \Delta U_k(t) - \frac{k^\epsilon}{(n^S)^2} = 0
		\end{equation*}
		Solving for $n^S$ yields:
		\begin{equation*}
		n^S = \sqrt{ \frac{k^\epsilon}{\beta\, \bar\theta_k(t)\, \Delta U_k(t)} }
		\end{equation*}
		Substituting $\bar\theta_k = k\,\Theta_k$ recovers \eqref{eq:n_closed_form_alt}.
		The solution must be constrained to $[\mathfrak{n}_{\text{min}}, 1]$, giving the result in (\ref{eq:n_closed_form}).
	\end{proof}
	
	For exposed individuals, under the current cost structure where transition to infectiousness is independent of contacts:
	\begin{equation}
	n_k^{E*}(t) = \arg\min_{n^E \in [\mathfrak{n}_{\text{min}},1]} \left[ f_k(n^E) \right] = 1 \quad \text{(since $f_k$ is decreasing in $n$)} \label{eq:opt_E}
	\end{equation}

\begin{remark}[Nash equilibrium value of $n_k^{E*}$ and extensions]
\label{rem:nE_nash}
In the baseline model, the HJB equation \eqref{eq:HJB_E} for an exposed individual minimizes only $f_k(n^E)$ (since the $\sigma(U_k^I - U_k^E)$ and $C_E$ terms do not depend on $n^E$). Because $f_k$ is decreasing, the Nash equilibrium strategy is $n_k^{E*}=1$ for all $t$: an exposed agent who has no remaining infection risk and faces only an isolation penalty will not voluntarily self-isolate.
When $\kappa>0$ (pre-symptomatic transmission), exposed individuals' contact rate $n_k^E$ enters the aggregate $\Theta_k$ through \eqref{eq:Theta_k_presym}. However, at Nash equilibrium each individual optimizes against the \emph{aggregate} field, not their own infinitesimal contribution. Therefore $n_k^{E*}=1$ continues to hold; the effect of $\kappa>0$ is felt indirectly through the aggregate pressure on susceptibles.

To obtain $n_k^{E*}<1$ as a Nash outcome, one must add an incentive that \emph{increases} the exposed agent's running cost when $n^E$ is large (high contacts).
A simple choice is a responsibility/compliance term that internalizes pre-symptomatic transmission externalities, for example
\begin{equation}
\eta_E\,\kappa\,\beta\,k\,n^E(t)\,\Theta_k(t),
\qquad \eta_E>0,
\label{eq:responsibility_cost}
\end{equation}
added to the running cost in state $E$.
Then the exposed HJB becomes
$\min_{n^E\in[\mathfrak n_{\min},1]}\big[\,\sigma(U_k^I-U_k^E)+C_E+f_k(n^E)+\eta_E\kappa\beta k n^E\Theta_k(t)\,\big]$,
which admits interior minimizers and yields $n_k^{E*}(t)<1$ when infection pressure is substantial.
\end{remark}

\subsection{Computing the Reproduction Number}
\label{sec:reproduction_number}

In degree-structured network models, the effective reproduction number is most naturally defined via a
next-generation operator acting on degree classes. In our setting (baseline: only $I$ transmits),
new exposures in class $k$ generated by infectious individuals in class $k'$ occur at instantaneous rate
\[
\text{new }E_k \ \propto\ \beta \,k\, n_k^S(t)\, S_k(t)\, G_{kk'}\, n_{k'}^I(t)\, I_{k'}(t).
\]
\begin{theorem}[Time-varying effective reproduction number as a spectral radius]
Define the $K\times K$ matrix
\begin{equation}
\mathcal{K}(t)_{kk'}
\;=\;
\frac{\beta}{\gamma}\, k\, n_k^S(t)\, S_k(t)\, G_{kk'}.
\label{eq:K_matrix}
\end{equation}
Then the time-varying effective reproduction number is
\begin{equation}
R_t \;=\; \rho\!\big(\mathcal{K}(t)\big),
\label{eq:Rt_spectral}
\end{equation}
where $\rho(\cdot)$ denotes the spectral radius.
\end{theorem}

\begin{proof}
We linearize the infection subsystem around the current state $(S_k(t))_{k=1}^K$ and consider a small perturbation in
infectious mass $(I_{k'})_{k'=1}^K$. A degree-$k'$ infectious individual remains infectious for mean time $1/\gamma$.
During that time, it produces new exposures in class $k$ along edges at rate
$\beta\,k\,n_k^S(t)\,S_k(t)\,G_{kk'}$ (using $n^I\equiv 1$).
Thus the expected number of secondary cases in class $k$ produced by one class-$k'$ infectious is exactly
$\mathcal{K}(t)_{kk'}$. The dominant growth/branching factor across degree classes is therefore the Perron root
$\rho(\mathcal{K}(t))$, which yields \eqref{eq:Rt_spectral}.
\end{proof}

\begin{corollary}[Basic reproduction number at disease-free equilibrium]
At the disease-free equilibrium $S_k(0)\approx 1$ and with baseline contact levels $n_k^S(0)=n_k^I(0)=1$,
\begin{equation}
R_0 \;=\; \rho\!\big(\mathcal{K}_0\big),
\qquad
(\mathcal{K}_0)_{kk'}=\frac{\beta}{\gamma}\,k\,G_{kk'}.
\label{eq:R0_general_corr}
\end{equation}
In the uncorrelated case $G_{kk'}=\frac{k'P(k')}{\langle k\rangle}$, $\mathcal{K}_0$ is rank-one and
\begin{equation}
R_0 \;=\; \frac{\beta}{\gamma}\,\frac{\langle k^2\rangle}{\langle k\rangle}.
\label{eq:R0_uncorr}
\end{equation}
\end{corollary}

\begin{remark}[Incubation rate $\sigma$ and $R_0$]
In the baseline SEIR model where only $I$ transmits, the incubation rate $\sigma$ affects timing (growth rate and peak),
but not the expected number of secondary infections produced by one infected individual, hence $R_0$ depends on
$\beta,\gamma$ (and network structure) but not on $\sigma$.
\end{remark}

\begin{remark}
For correlated networks, compute $R_t$ by taking the largest eigenvalue of $\mathcal{K}(t)$ (power iteration is sufficient).
For uncorrelated networks, the spectral radius reduces to the scalar
\[
R_t=\frac{\beta}{\gamma}\frac{1}{\langle k\rangle}\sum_k k^2 P(k)\,n_k^S(t)\,S_k(t).
\]
\end{remark}
	
\section{Nash Equilibrium Characterization}
\label{sec:nash_equilibrium}
	
\subsection{Complete SEIR-MFG System}
	
The MFG is characterized by the following coupled forward-backward system.

Forward (Kolmogorov) equations:
	\begin{equation}
	\begin{split}
	&\dot{S}_k(t) = -\beta k n_k^S(t) \Theta_k(t) S_k(t)  \\
	&\dot{E}_k(t) = \beta k n_k^S(t) \Theta_k(t) S_k(t) - \sigma E_k(t)  \\
	&\dot{I}_k(t) = \sigma E_k(t) - \gamma I_k(t)  \\
	&\dot{R}_k(t) = \gamma I_k(t)  \\
	&\Theta_k(t) = \sum_{k'} G_{kk'}\, I_{k'}(t) 
    \end{split}
	\end{equation}

Backward (HJB) equations:
    \begin{equation}
	\begin{split}
	&-\frac{dU_k^S}{dt} = \beta k n_k^S(t) \Theta_k(t) \left( r_I + U_k^E(t) - U_k^S(t) \right) + f_k(n_k^S(t))  \\
	&-\frac{dU_k^E}{dt} = \sigma \left( U_k^I(t) - U_k^E(t) \right) + C_E + f_k(n_k^E(t)) \\
	&-\frac{dU_k^I}{dt} = \gamma \left( U_k^R(t) - U_k^I(t) \right) + C_I \\
	&-\frac{dU_k^R}{dt} = 0. 
\end{split}
	\end{equation}
Consistency conditions:
    \begin{equation}
	n_k^S(t) = n_k^{S*}(t), \quad n_k^E(t) = n_k^{E*}(t) 
	\label{eq:consistency}
	\end{equation}
	with initial conditions $S_k(0) \approx 1$, $E_k(0) = \varepsilon_k$, $I_k(0) = R_k(0) = 0$, and terminal conditions $U_k^x(T) = \Psi(x)$.
	The consistency condition (\ref{eq:consistency}) ensures that the effort parameters appearing in the population dynamics (FP) match those chosen by individuals optimizing against those same dynamics (HJB). This fixed-point condition defines the Nash equilibrium.
	
\subsection{Existence and Uniqueness}
	
\begin{theorem}[Existence of Nash Equilibrium]
\label{thm:existence}
    Consider the SEIR-MFG system. Assume:
    \begin{enumerate}
        \item The social cost function $f_k(n)$ is strictly convex, twice continuously differentiable on $[\mathfrak{n}_{\text{min}}, 1]$, with $f_k'(n) < 0$ and $f_k''(n) > 0$
        \item Model parameters satisfy $0 < \beta, \sigma, \gamma < \infty$
        \item Initial conditions: $S_k(0)>0$, $E_k(0),I_k(0),R_k(0)\ge 0$, with $S_k(0)+E_k(0)+I_k(0)+R_k(0)=1$, and $\sum_k(E_k(0)+I_k(0))>0$ (nontrivial outbreak)
        \item Degree distribution $P(k)$ has finite support $k \in \{1,2,\dots,K\}$
    \end{enumerate}
    Then there exists a Nash equilibrium on $[0,T]$.
\end{theorem}

\begin{proof}
Since we have finitely many degree classes $k\in\{1,\ldots,K\}$ and finite state space $\{S,E,I,R\}$, the state of the system is characterised by $4K$ continuous functions. We work in the Banach space
\[
\mathcal{X} = C([0,T])^K \times C([0,T])^K
\]
of pairs of continuous policies $(n^S(\cdot), n^E(\cdot))$ with $n^S_k, n^E_k : [0,T]\to[\mathfrak{n}_{\min},1]$, equipped with the sup-norm $\|\cdot\|_\infty$. Define
\[
\mathcal{A} = \left\{ (n^S,n^E)\in\mathcal{X} : n^S_k(t),\, n^E_k(t)\in[\mathfrak{n}_{\min},1]\ \forall k,t \right\},
\]
which is a closed, bounded, convex subset of $\mathcal{X}$. We continue in several steps.

\textbf{Step 1 (Forward map is Lipschitz).} Given $(n^S, n^E)\in\mathcal{A}$, the forward system \eqref{eq:S_dot}--\eqref{eq:R_dot} is a system of ODEs with right-hand sides that are \emph{uniformly Lipschitz} in the state (since $n^S_k\in[\mathfrak{n}_{\min},1]$ and all quantities are bounded). The Cauchy--Lipschitz theorem yields a unique solution $\mu[n]\in C^1([0,T])^{4K}$, bounded in $[\,0,1\,]^{4K}$.

\textbf{Step 2 (Backward map is Lipschitz).} Given $\mu[n]$, the HJB system \eqref{eq:HJB_S}--\eqref{eq:HJB_R} is a backward ODE with Lipschitz right-hand side (using the bounds on $\mu$ and $n$). The terminal-value problem has a unique solution $U[n]\in C^1([0,T])^{4K}$, uniformly bounded by $C_U := T(f_k({\mathfrak{n}_{\min}})+C_E+C_I)+\max|\Psi|$.

\textbf{Step 3 (Best-response map).} Define $\Phi:\mathcal{A}\to\mathcal{A}$ by
\[
[\Phi(n)]_k^S(t) =
\begin{cases}
1, & \bar\theta_k(t)\,\Delta U_k(t)\le 0,\\[1mm]
\text{proj}_{[\mathfrak{n}_{\min},1]}\!\left(\sqrt{\dfrac{k^{\epsilon}}{\beta\,\bar\theta_k(t)\,\Delta U_k(t)}}\right),
& \bar\theta_k(t)\,\Delta U_k(t)>0,
\end{cases}
\qquad [\Phi(n)]_k^E(t)\equiv 1.
\]
where $\bar\theta_k, \Delta U_k$ are computed from Steps 1--2.

\textbf{Step 4 (Compactness via Arzel`a--Ascoli).} The image $\Phi(\mathcal{A})$ is equicontinuous.
Indeed, $(\bar\theta_k(t))_k$ and $(\Delta U_k(t))_k$ are Lipschitz in $t$ (solutions of Lipschitz ODEs).
Define the clipped map
\[
g(z):=\min\{1,\max\{\mathfrak n_{\min},z^{-1/2}\}\},\qquad z>0,
\]
and interpret it piecewise: $g(z)=1$ for $z\le 1$, $g(z)=z^{-1/2}$ for $1<z<1/\mathfrak n_{\min}^2$, and $g(z)=\mathfrak n_{\min}$ for $z\ge 1/\mathfrak n_{\min}^2$.
On each region $g$ is Lipschitz, hence globally Lipschitz on $(0,\infty)$.
Applying $g$ to $z_k(t):=\beta\,\bar\theta_k(t)\,\Delta U_k(t)/k^{\epsilon}$ (with the convention $n^{S*}=1$ if $z_k(t)\le 0$) yields a uniform modulus of continuity for $[\Phi(n)]_k^S(\cdot)$.
Therefore $\Phi(\mathcal{A})$ is relatively compact in $\mathcal{X}$ by Arzel`a--Ascoli.

\textbf{Step 5 (Continuity of $\Phi$).} If $n^{(j)}\to n$ in $\mathcal{A}$ (uniform convergence), then by continuous dependence of the ODE solutions (Gr\"onwall), $\mu[n^{(j)}]\to\mu[n]$ and $U[n^{(j)}]\to U[n]$ uniformly, so $\Phi(n^{(j)})\to\Phi(n)$ uniformly.

 $\mathcal{A}$ is a closed, bounded, convex subset of the infinite-dimensional space $\mathcal{X}=C([0,T])^{2K}$; it is \emph{not} compact (closed bounded convex sets in infinite-dimensional Banach spaces are not compact in general). What Steps 4--5 give instead is: $\Phi(\mathcal{A})$ is \emph{relatively compact} (by Arzelà--Ascoli) and $\Phi$ is continuous. By the Schauder fixed-point theorem in this form---a continuous map on a closed convex set whose image is relatively compact has a fixed point---$\Phi$ has a fixed point $n^*\in\mathcal{A}$, and $(\mu[n^*], n^*, U[n^*])$ is a Nash equilibrium.
\end{proof}

	\begin{theorem}[Uniqueness under Monotonicity Conditions]
    Assume in addition to the conditions of Theorem~\ref{thm:existence} that:
    \begin{enumerate}
        \item The Hamiltonian $H_k^S(n^S, U, \Theta_k)$ satisfies the displacement monotonicity condition: For any two sets of value functions $U^1, U^2 \in \mathcal{V}$ and corresponding infection pressures $\Theta^1, \Theta^2$,
        \begin{equation*}
            \sum_{k=1}^K \int_0^T \left[ H_k^S(n_k^{S*1}, U_k^1, \Theta^1) - H_k^S(n_k^{S*2}, U_k^1, \Theta^1) - H_k^S(n_k^{S*1}, U_k^2, \Theta^2) + H_k^S(n_k^{S*2}, U_k^2, \Theta^2) \right] dt \geq 0
        \end{equation*}
        with equality only when $n^{S*1} = n^{S*2}$ almost everywhere.
        
        \item The social cost $f_k(n)$ is strongly convex: $f_k''(n) \geq \alpha > 0$ for all $n \in [\mathfrak{n}_{\text{min}}, 1]$.
    \end{enumerate}
    Then the Nash equilibrium is unique.
\end{theorem}

\begin{proof}
    Suppose there exist two distinct Nash equilibria $(\mu^1, n^1, U^1)$ and $(\mu^2, n^2, U^2)$. Denote differences by $\delta S_k = S_k^1 - S_k^2$, $\delta n_k^S = n_k^{S1} - n_k^{S2}$, $\delta U_k^S = U_k^{S1} - U_k^{S2}$, etc. From the optimality conditions:
    \begin{align*}
        n_k^{S1}(t) &= \argmin_{n^S} H_k^S(n^S, U_k^1(t), \Theta^1(t)) \\
        n_k^{S2}(t) &= \argmin_{n^S} H_k^S(n^S, U_k^2(t), \Theta^2(t)).
    \end{align*}
    By optimality, we have for almost every $t$:
    \begin{align}
        H_k^S(n_k^{S1}, U_k^1, \Theta^1) - H_k^S(n_k^{S2}, U_k^1, \Theta^1) &\leq 0 \label{eq:opt1} \\
        H_k^S(n_k^{S2}, U_k^2, \Theta^2) - H_k^S(n_k^{S1}, U_k^2, \Theta^2) &\leq 0. \label{eq:opt2}
    \end{align}
    Adding \eqref{eq:opt1} and \eqref{eq:opt2}:
    \begin{equation}
        \left[ H_k^S(n_k^{S1}, U_k^1, \Theta^1) - H_k^S(n_k^{S2}, U_k^1, \Theta^1) \right] + \left[ H_k^S(n_k^{S2}, U_k^2, \Theta^2) - H_k^S(n_k^{S1}, U_k^2, \Theta^2) \right] \leq 0. \label{eq:sum_ineq}
    \end{equation}
 By strong convexity of $f_k$:
    \begin{align*}
        &H_k^S(n_k^{S1}, U_k^1, \Theta^1) - H_k^S(n_k^{S2}, U_k^1, \Theta^1) \\
        &\geq \beta k \Theta^1 (r_I + U_k^{E1} - U_k^{S1})(n_k^{S1} - n_k^{S2}) + f_k'(n_k^{S2})(n_k^{S1} - n_k^{S2}) + \frac{\alpha}{2}(n_k^{S1} - n_k^{S2})^2.
    \end{align*}
    Since $n_k^{S2}$ minimizes $H_k^S(\cdot, U_k^2, \Theta^2)$, we have $f_k'(n_k^{S2}) = -\beta k \Theta^2 (r_I + U_k^{E2} - U_k^{S2})$. Thus:
    \begin{align*}
        &\text{LHS of \eqref{eq:sum_ineq}} \\
        &\geq \beta k \left[ \Theta^1 (r_I + U_k^{E1} - U_k^{S1}) - \Theta^2 (r_I + U_k^{E2} - U_k^{S2}) \right] (n_k^{S1} - n_k^{S2}) + \alpha (n_k^{S1} - n_k^{S2})^2.
    \end{align*}
From the forward equations, we have
    \begin{align*}
        \frac{d}{dt}(\delta S_k) &\leq C_1 (|\delta n_k^S| + |\delta \Theta_k| + |\delta S_k|).
    \end{align*}
    By Grönwall's inequality: $\|\delta S_k\|_{L^\infty} \leq C_2 \int_0^T (|\delta n_k^S| + |\delta \Theta_k|) dt$.
    For the value functions, from the HJB equations:
    \begin{align*}
        &\left|\frac{d}{dt}(\delta U_k^S)\right|
        \leq C_3 (|\delta n_k^S| + |\delta U_k^S| + |\delta U_k^E| + |\delta \Theta_k|).
    \end{align*}
    Integrating backward gives: $\|\delta U_k^S\|_{L^\infty} \leq C_4 \int_0^T (|\delta n_k^S| + |\delta \Theta_k|) dt$.
     The infection pressure difference satisfies that
    \begin{equation*}
        |\delta \Theta_k(t)| \leq C_5 \sum_{k'} |\delta I_{k'}(t)|.
    \end{equation*}
    and from the $I_k$ equation: $|\delta I_k(t)| \leq C_6 \int_0^t (|\delta E_k(s)| + |\delta I_k(s)|) ds$.
    Combining all estimates with the monotonicity condition yields:
    \begin{equation*}
        \alpha \sum_{k=1}^K \int_0^T |\delta n_k^S(t)|^2 dt \leq C_7 \sum_{k=1}^K \int_0^T |\delta n_k^S(t)| \cdot \left( \int_0^T |\delta n_k^S(s)| ds \right) dt.
    \end{equation*}
    Applying Young's inequality and Grönwall's lemma gives $\delta n_k^S = 0$ almost everywhere. Then $\delta \mu_k = 0$ and $\delta U_k = 0$ follow from the differential equations, establishing uniqueness.
\end{proof}
	
\subsection{Regularity of Solutions}
	
	\begin{lemma}[Boundedness and regularity]
		Under the assumptions of Theorem 1, the solutions satisfy:
		\begin{enumerate}
			\item $S_k(t), E_k(t), I_k(t), R_k(t) \in [0,1]$ for all $t \in [0,T]$ and $\sum_{x \in \{S,E,I,R\}} X_k(t) = 1$.
			\item $U_k^S(t), U_k^E(t), U_k^I(t), U_k^R(t)$ are uniformly bounded on $[0,T]$.
			\item $n_k^{S*}(t), n_k^{E*}(t) \in [\mathfrak{n}_{\text{min}}, 1]$ and are measurable functions.
		\end{enumerate}
	\end{lemma}
	
	\begin{proof}
		We see that the right-hand sides of (\ref{eq:S_dot})-(\ref{eq:R_dot}) sum to zero:
			\begin{align*}
			\frac{d}{dt}\sum_{x \in \{S,E,I,R\}} X_k(t) = -\beta k n_k^S \Theta_k S_k + \beta k n_k^S \Theta_k S_k - \sigma E_k + \sigma E_k - \gamma I_k + \gamma I_k = 0,
			\end{align*}
			so $\sum_{x} X_k(t) = \sum_{x} X_k(0) = 1$. Non-negativity follows from the structure of the equations.
	For the value functions:
			\begin{align*}
			|U_k^S(t)| &\leq \int_t^T \left( \max_{n \in [\mathfrak{n}_{\text{min}},1]} |f_k(n)| + C_E + C_I \right) ds + |\Psi(S)| \\
			&\leq T \left( \max_{n} |f_k(n)| + C_E + C_I \right) + \max_{x} |\Psi(x)|.
			\end{align*}
			Similar bounds hold for $U_k^E, U_k^I, U_k^R$.
			 The optimal controls are defined as minimizers of continuous functions over the compact set $[\mathfrak{n}_{\text{min}}, 1]$, so they exist and are measurable by measurable selection theorems.

	\end{proof}

	\begin{theorem}[Continuous Dependence on Parameters]
    For two parameter sets $\mathcal{P}^1, \mathcal{P}^2$ with $\|\mathcal{P}^1 - \mathcal{P}^2\| \leq \delta$, the corresponding equilibria satisfy:
    \begin{equation*}
        \sup_{t \in [0,T]} \sum_{k=1}^K \left( |\mu_k^1(t) - \mu_k^2(t)| + |U_k^1(t) - U_k^2(t)| \right) + \int_0^T |n_k^1(t) - n_k^2(t)| dt \leq C(T, L) \delta,
    \end{equation*}
    where $C(T, L)$ depends on time horizon $T$ and Lipschitz constants of the dynamics.
\end{theorem}

\begin{proof}
    Let $(\mu^1, n^1, U^1)$ and $(\mu^2, n^2, U^2)$ be equilibria for parameters $\mathcal{P}^1$ and $\mathcal{P}^2$. Denote differences by $\delta \mu_k = \mu_k^1 - \mu_k^2$, $\delta n_k = n_k^1 - n_k^2$, $\delta U_k = U_k^1 - U_k^2$, $\delta \Theta_k = \Theta_k^1 - \Theta_k^2$.
 From the $S_k$ equation, decomposing the difference:
    \begin{align*}
        \left|\frac{d}{dt}(\delta S_k)\right| &\leq L_1 (|\delta n_k^S| + |\delta \Theta_k| + |\delta S_k| + \delta),
    \end{align*}
    since all quantities are bounded ($|n_k^S| \leq 1$, $|S_k| \leq 1$, $|\Theta_k| \leq 1$). For $\delta \Theta_k$:
    \begin{equation*}
        |\delta \Theta_k(t)| \leq \frac{K \max k}{\langle k \rangle} \sum_{k'} |\delta I_{k'}|.
    \end{equation*}
From the HJB for $U_k^S$:
    \begin{align*}
        \left|\frac{d}{dt}(\delta U_k^S)\right|
        &\leq L_2 (|\delta n_k^S| + |\delta \Theta_k| + |\delta U_k^S| + |\delta U_k^E| + \delta).
    \end{align*}
When solutions are interior, the first-order optimality conditions give $f_k'(n_k^{Si}) = -\beta k \Theta^i (r_I + U_k^{Ei} - U_k^{Si})$ for $i=1,2$. By the mean value theorem and strong convexity:
    \begin{equation*}
        |\delta n_k^S| \leq \frac{L_3}{\alpha} (|\delta U_k^S| + |\delta U_k^E| + |\delta \Theta_k| + \delta).
    \end{equation*}
  Finally, we define:
    \begin{equation*}
        \Phi(t) = \sum_{k=1}^K \left( |\delta S_k(t)| + |\delta E_k(t)| + |\delta I_k(t)| + |\delta U_k^S(t)| + |\delta U_k^E(t)| \right).
    \end{equation*}
    Combining the estimates from Steps 1--3:
    \begin{equation*}
        \frac{d\Phi}{dt} \leq L_5 (\Phi + \delta).
    \end{equation*}
    By Grönwall's inequality with $\Phi(0) \leq C_0\delta$:
    \begin{equation*}
        \Phi(t) \leq C(T, L) \delta, \qquad \int_0^T |\delta n_k^S(t)| dt \leq C'(T, L) \delta,
    \end{equation*}
    establishing the claimed continuous dependence.
\end{proof}

\section{Special Cases and Analytical Insights}
\label{sec:special_cases}

\subsection{Strategic and Biological Roles of the Incubation Rate}

The introduction of the exposed compartment reshapes the epidemic game through the value gap $\Delta U_k(t) = r_I + U_k^E(t) - U_k^S(t)$, the marginal cost of infection. As the following remark makes precise, $\sigma$ acts on this gap through two opposing channels.

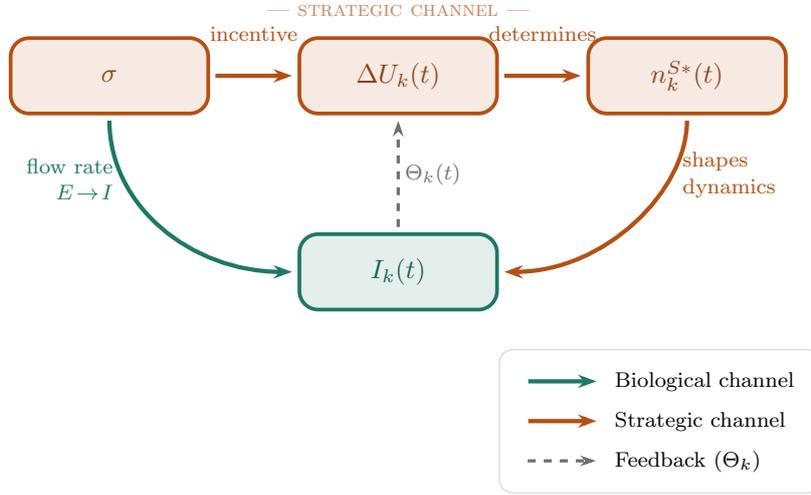
\begin{figure}[h]
    \centering
     
    \definecolor{stratcol}{RGB}{180,80,20}
    \definecolor{biocol}{RGB}{30,120,100}
    \definecolor{feedcol}{RGB}{110,110,110}
    \begin{tikzpicture}[
      mainbox/.style 2 args={
        draw=#1, fill=#1!12, rounded corners=7pt,
        minimum width=2.6cm, minimum height=1.0cm,
        font=\small\bfseries, text=#1!80!black,
        line width=1.4pt, align=center
      },
      strarr/.style={-{Stealth[length=7pt,width=5pt]}, line width=1.6pt,
                     color=stratcol, shorten >=2pt, shorten <=2pt},
      bioarr/.style={-{Stealth[length=7pt,width=5pt]}, line width=1.6pt,
                     color=biocol, shorten >=2pt, shorten <=2pt},
      feedarr/.style={-{Stealth[length=6pt,width=4pt]}, line width=1.2pt,
                      color=feedcol, dashed, shorten >=2pt, shorten <=2pt},
      arrlbl/.style={font=\scriptsize, inner sep=2pt}
    ]
      \node[mainbox={stratcol}{orange}] (sigma)  at (0,0)      {$\sigma$};
      \node[mainbox={stratcol}{orange}] (dU)     at (3.8,0)    {$\Delta U_k(t)$};
      \node[mainbox={stratcol}{orange}] (nstar)  at (7.6,0)    {$n_k^{S*}(t)$};
      \node[mainbox={biocol}{teal}]     (Ik)     at (3.8,-2.6){$I_k(t)$};
      
      \draw[strarr] (sigma) -- (dU);
      \node[arrlbl, color=stratcol, yshift=16pt] at ($(sigma)!0.5!(dU)$) {incentive};
      \draw[strarr] (dU) -- (nstar);
      \node[arrlbl, color=stratcol, yshift=16pt] at ($(dU)!0.5!(nstar)$) {determines};
      \draw[strarr] (nstar.south) .. controls +(0,-1.0) and +(1.4,0) .. (Ik.east);

      \node[arrlbl, color=stratcol, right=4pt, xshift=19pt] at ($(nstar.south)!0.38!(Ik.east)$)
          {\begin{tabular}{@{}l@{}}shapes\\dynamics\end{tabular}};
          
      \draw[bioarr] (sigma.south) .. controls +(0,-1.2) and +(-1.4,0) .. (Ik.west);
      \node[arrlbl, color=biocol, left=4pt, xshift=-22pt] at ($(sigma.south)!0.42!(Ik.west)$)
          {\begin{tabular}{@{}r@{}}flow rate\\$E\!\to\!I$\end{tabular}};
          
      \draw[feedarr] (Ik.north) .. controls +(0,0.8) and +(0,-0.8) .. (dU.south)
          node[arrlbl, right, color=feedcol, pos=0.5] {$\Theta_k(t)$};
          
      \node[font=\scriptsize\scshape, color=stratcol!70, above=4pt of dU]
          {--- strategic channel ---};
          
      \node[draw=gray!35, fill=white, rounded corners=5pt,
            inner xsep=10pt, inner ysep=7pt,
            below right=0.5cm and 0.0cm of Ik,
            font=\scriptsize, align=left] {%
          \begin{tabular}{@{}l@{\hspace{6pt}}l@{}}
            \tikz[baseline=-0.6ex]{\draw[bioarr, shorten >=0, shorten <=0] (0,0)--(0.9,0);} &
              Biological channel \\[5pt]
            \tikz[baseline=-0.6ex]{\draw[strarr, shorten >=0, shorten <=0] (0,0)--(0.9,0);} &
              Strategic channel \\[5pt]
            \tikz[baseline=-0.6ex]{\draw[feedarr, shorten >=0, shorten <=0] (0,0)--(0.9,0);} &
              Feedback ($\Theta_k$)
          \end{tabular}%
      };
    \end{tikzpicture}
    \caption{\textbf{The Dual Role of Incubation ($\sigma$).} The incubation rate $\sigma$ affects the epidemic through two distinct channels. \textbf{Biologically} (bottom arrow), it controls the flow rate from exposed to infectious ($\dot{I} \propto \sigma E$). \textbf{Strategically} (top arrow), it determines the immediacy of the infection cost, altering the value gap $\Delta U_k$ and thus the contact effort $n_k^{S*}$.}
    \label{fig:causal_chain}
\end{figure}
\subsection{Degree Scaling of Optimal Effort}

For homogeneous networks where all nodes have degree $k$, the model yields clear analytical insights.

\begin{proposition}[Degree scaling of optimal effort]
\label{prop:effort_scaling}
Assume the uncorrelated reduction \eqref{eq:Theta_uncorr} and suppose the minimizer in the susceptible HJB is interior
(i.e., $\mathfrak n_{\min}<n_k^{S*}(t)<1$). With $f_k(n)=k^\epsilon\left(\frac1n-1\right)$ we have
\begin{equation}
n_k^{S*}(t)
=
\left(\frac{k^{\epsilon-1}}{\beta\,\Theta(t)\,\Delta U_k(t)}\right)^{1/2},
\qquad
\Delta U_k(t)=r_I+U_k^E(t)-U_k^S(t).
\label{eq:nk_scaling_degree}
\end{equation}
Consequently:
\begin{itemize}
\item if $\epsilon<1$, then $n_k^{S*}(t)$ decreases with $k$ (high-degree agents take \emph{more} precaution);
\item if $\epsilon=1$, then $n_k^{S*}(t)=(\beta\,\Theta(t)\,\Delta U_k(t))^{-1/2}$ is \emph{independent of $k$};
\item if $\epsilon>1$, then $n_k^{S*}(t)$ increases with $k$ (high-degree agents take \emph{less} precaution).
\end{itemize}
Degree thresholds appear only through saturation at the box constraints $[\mathfrak n_{\min},1]$.
\end{proposition}

\begin{proof}
With \eqref{eq:lambda_k} and the uncorrelated reduction, the susceptible Hamiltonian is
\[
H_k^S(n)=\beta\,k\,n\,\Theta(t)\,\Delta U_k(t)+k^\epsilon\Big(\frac1n-1\Big).
\]
For an interior minimizer, $0= \partial_n H_k^S(n)=\beta k\Theta \Delta U_k - k^\epsilon/n^2$, hence
$n^2=k^{\epsilon-1}/(\beta\Theta \Delta U_k)$, which gives \eqref{eq:nk_scaling_degree}.
The monotonicity conclusions follow immediately from the exponent $(\epsilon-1)/2$.
\end{proof}

\begin{remark}[Minimized cost rate is linear in $k$ for $\epsilon=1$]
\label{rem:Hstar}
Substituting the interior minimizer back for the case $\epsilon=1$ gives
\[
H_k^*(t) = k\!\left(2\sqrt{\beta\,\Theta(t)\,\Delta U_k(t)}-1\right).
\]
This is \emph{linear} in $k$ (assuming $\Delta U_k$ does not depend on $k$, which holds at zeroth order when degree effects on value functions are negligible). Therefore no interior critical degree minimizes or maximizes the cost rate among interior-solution agents; differences across degree classes arise purely through the saturation behaviour at the constraints $n^{S*}=\mathfrak{n}_{\min}$ or $n^{S*}=1$.
\end{remark}

\subsection{Comparison with SIR-MFG}

\begin{theorem}[SEIR vs.\ SIR-MFG Comparison]
\label{thm:comparison}
    Let $(\mu^{\text{SEIR}}, n^{\text{SEIR}}, U^{\text{SEIR}})$ be a SEIR-MFG equilibrium with incubation rate $\sigma$, and $(\mu^{\text{SIR}}, n^{\text{SIR}}, U^{\text{SIR}})$ be a SIR-MFG equilibrium with equivalent parameters $\beta, \gamma, r_I, C_I$ and initial conditions. Assume $\sigma, \gamma > 0$ and the epidemics are not trivial ($I(0) > 0$). Then:
    \begin{enumerate}
        \item \textbf{Early-phase growth:} For a fixed contact level $\bar{n}$, the linearized early growth rate satisfies $\lambda_{\mathrm{SEIR}} < \lambda_{\mathrm{SIR}}$ for all finite $\sigma$ (Proposition~\ref{prop:early_growth}). This comparison concerns the initial exponential phase and does not by itself determine peak timing in the full nonlinear dynamics.
        \item \textbf{Attenuated Response (numerical evidence):} Simulations show $\max_t |1 - n^{\mathrm{SEIR}}(t)| \leq \max_t |1 - n^{\mathrm{SIR}}(t)|$. See Remark~\ref{rem:attenuated} for a discussion of why a clean analytical proof is elusive.
        \item \textbf{Larger Final Size (conditional):} If the attenuated response holds, then $R^{\mathrm{SEIR}}(\infty) \geq R^{\mathrm{SIR}}(\infty)$ follows from the final-size relation \eqref{eq:final_size}.
    \end{enumerate}
\end{theorem}

\begin{proof}
    \textbf{Early-phase growth.} See Proposition~\ref{prop:early_growth}. The inequality $\lambda_{\mathrm{SEIR}} < \lambda_{\mathrm{SIR}}$ follows from the characteristic polynomial of the SEIR linearization for fixed contact level $\bar n$.

    \textbf{Conditional Final Size.} Integrating $\dot{S}_k = -\beta k n^S \Theta_k S_k$ and using $\int_0^\infty I\,dt = R(\infty)/\gamma$ (which follows from $\dot{R}=\gamma I$) gives the final-size relation for both SEIR and SIR:
    \begin{equation}
    \ln\frac{S(0)}{S(\infty)} = \frac{\beta}{\gamma} \int_0^\infty n^S(t)\,I(t)\,dt.
    \label{eq:final_size}
    \end{equation}
    If $n^{\mathrm{SEIR}}(t) \geq n^{\mathrm{SIR}}(t)$ for all $t$ (the attenuated response claim, Part~2), and $I^{\mathrm{SEIR}}(t)$ is more temporally dispersed (Part~1), then the right-hand integral is larger for SEIR, giving $R^{\mathrm{SEIR}}(\infty) \geq R^{\mathrm{SIR}}(\infty)$.
\end{proof}

\begin{remark}[Why the attenuated response is difficult to prove analytically]
\label{rem:attenuated}
A natural attempt is to compare $\Delta U^{\text{SEIR}} = r_I + U^E_{\text{SEIR}} - U^S_{\text{SEIR}}$ with $\Delta U^{\text{SIR}} = r_I + U^I_{\text{SIR}} - U^S_{\text{SIR}}$.

However, the na\"ive inequality $U^I_{\text{SIR}} \geq U^E_{\text{SEIR}}$ is \emph{not} obviously true, nor is its reverse. In fact, within the SEIR model itself, since being in state $E$ means the agent must still pass through the $E$-sojourn (incurring cost $C_E$) before reaching $I$, the ordering $U^E_{\text{SEIR}} \geq U^I_{\text{SEIR}}$ holds. But this is a within-model comparison, whereas Part~2 requires comparing value functions across two \emph{different} games (SEIR-MFG vs SIR-MFG) with different population dynamics. A rigorous proof would require tracking how the equilibrium trajectories $I^{\text{SEIR}}(t)$ and $I^{\text{SIR}}(t)$ differ and propagating those differences back through the HJB. This is left for future work; numerical evidence in Section~\ref{sec:numerical} consistently supports the claim.

Note also that there is no discount factor in the HJB \eqref{eq:HJB_S}--\eqref{eq:HJB_R}; the strategic delay effect arises from the finite horizon and the temporal structure of costs, not from exponential discounting.
\end{remark}

\subsection{Early-Phase Growth Rate Comparison}

\begin{proposition}[Early-phase growth rate: SEIR is slower than SIR for the same contact level]
\label{prop:early_growth}
Fix a constant contact level $n^S\equiv \bar n\in(0,1]$ and consider the homogeneous mean-field dynamics with $S(t)\approx 1$.
For SIR, the linearized growth rate is $\lambda_{\mathrm{SIR}}=\beta \bar n-\gamma$.
For SEIR (baseline: only $I$ transmits), the linearized growth rate is the largest root of
\[
\lambda^2+(\sigma+\gamma)\lambda+\sigma(\gamma-\beta \bar n)=0,
\]
i.e.
\[
\lambda_{\mathrm{SEIR}}
=
\frac{-(\sigma+\gamma)+\sqrt{(\sigma+\gamma)^2+4\sigma(\beta\bar n-\gamma)}}{2}.
\]
Whenever $\beta\bar n>\gamma$ (supercritical regime), one has
\[
0<\lambda_{\mathrm{SEIR}}<\lambda_{\mathrm{SIR}}.
\]
\end{proposition}

\begin{proof}
The SIR formula is immediate from $\dot I=(\beta\bar n-\gamma)I$.
For SEIR, linearizing gives $\dot E=\beta\bar n I-\sigma E$ and $\dot I=\sigma E-\gamma I$, so
$\frac{d}{dt}\binom{E}{I}=A\binom{E}{I}$ with
$A=\begin{pmatrix}-\sigma & \beta\bar n\\ \sigma & -\gamma\end{pmatrix}$.
The characteristic polynomial is exactly $\lambda^2+(\sigma+\gamma)\lambda+\sigma(\gamma-\beta\bar n)$.
In the supercritical regime $\beta\bar n>\gamma$, the Perron root is positive.
A direct comparison of the explicit formulas yields $\lambda_{\mathrm{SEIR}}<\beta\bar n-\gamma$ for all finite $\sigma$.
\end{proof}

\subsection{Consequences of Incubation for Strategic Incentives}

\begin{proposition}[Strategic delay in precaution onset]
\label{prop:delay_onset}
In the supercritical regime, the optimal susceptible effort satisfies $n_k^{S*}(t) = 1$ on an initial interval $[0, t^*_\sigma)$. As $\sigma$ decreases (longer incubation), $t^*_\sigma$ increases.
\end{proposition}

\begin{proof}
By \eqref{eq:nk_scaling_degree}, the interior solution $n_k^{S*} < 1$ requires $\beta\,\Theta(t)\,\Delta U_k(t) > k^{\epsilon-1}$. At $t=0$ with $I(0)=\varepsilon \ll 1$, we have $\Theta(0) = O(\varepsilon)$, so $\beta\,\Theta(0)\,\Delta U_k(0) < k^{\epsilon-1}$ for small $\varepsilon$. Hence $n_k^{S*}(0) = 1$. The threshold $t^*_\sigma$ is defined by $\beta\,\Theta(t^*_\sigma)\,\Delta U_k(t^*_\sigma) = k^{\epsilon-1}$. Under the baseline definition $\Delta U_k(t)=r_I+U_k^E(t)-U_k^S(t)$, the gap is driven by the epidemic trajectory through the coupled backward HJB system. For smaller $\sigma$, $I_k(t)$ builds up more slowly (the mean $E$-sojourn time $\sigma^{-1}$ must elapse before $I$ grows), keeping $\Theta(t)$ smaller at early times. The product $\beta\,\Theta\,\Delta U_k$ therefore remains below $k^{\epsilon-1}$ for longer, so $t^*_\sigma$ increases as $\sigma$ decreases.
\end{proof}

\subsection{The SIR Limit ($\sigma\to\infty$)}

\begin{lemma}[Convergence to SIR-MFG]
\label{lem:sir_limit}
Fix all parameters except $\sigma$ and suppose $C_E = 0$, $\Psi\equiv 0$. As $\sigma\to\infty$:
\begin{enumerate}[label=\emph{(\arabic*)}]
\item $E_k(t)\to 0$ uniformly on $[0,T]$.
\item $\|U_k^{E,\sigma} - U_k^{I,\sigma}\|_\infty \to 0$.
\item $\Delta U_k^\sigma(t) \to \Delta U_k^{\mathrm{SIR}}(t) = r_I + U_k^{I,\infty}(t) - U_k^{S,\infty}(t)$ uniformly.
\item The SEIR-MFG Nash equilibrium converges to the SIR-MFG equilibrium of \cite{Bremaud2025} at rate $O(\sigma^{-1})$.
\end{enumerate}
\end{lemma}

\begin{proof}
From \eqref{eq:E_dot}, $\dot E_k = \lambda_k S_k - \sigma E_k$. By variation of constants, $E_k = O(\sigma^{-1})$ uniformly once $\sigma > \|\lambda_k\|_\infty$, establishing~(1). Set $w_k = U_k^I - U_k^E$. Then $w_k(T)=0$ and
\[
\dot w_k = \sigma w_k + \bigl[\gamma(U_k^R - U_k^I) + C_I\bigr].
\]
The source is bounded uniformly in $\sigma$ by the boundedness lemma. Integrating backward from $T$:
\[
|w_k(t)| \;\leq\; \frac{M}{\sigma}(1-e^{-\sigma(T-t)}) \;\leq\; \frac{M}{\sigma},
\]
establishing~(2). Then $\Delta U_k^\sigma = r_I + U_k^{E,\sigma} - U_k^{S,\sigma} = r_I + U_k^{I,\sigma} - U_k^{S,\sigma} + O(\sigma^{-1})$. The limiting gap $r_I + U_k^{I,\infty} - U_k^{S,\infty}$ is the SIR-MFG value gap $\Delta U_k^{\mathrm{SIR}}$, establishing~(3). Since $\Delta U_k^\sigma \to \Delta U_k^{\mathrm{SIR}}$ uniformly, the optimal policy \eqref{eq:nk_scaling_degree} converges $n_k^{S,\sigma} \to n_k^{S,\mathrm{SIR}}$; combined with~(1), the forward dynamics converge to the SIR system, establishing~(4).
\end{proof}

\section{Numerical Simulations}
\label{sec:numerical}

We solve the coupled forward-backward SEIR-MFG system by a standard forward-backward sweep (FBS) iteration, alternating between (i) simulating the population dynamics forward in time under a candidate policy, (ii) solving the HJB system backward in time under the resulting epidemic trajectory, and (iii) updating the policy by pointwise minimization of the Hamiltonians. A relaxation parameter $\omega\in(0,1]$ prevents oscillations.

\subsection{Iterative Forward-Backward Sweep Algorithm}

Fix a time grid $0=t_0<t_1<\cdots<t_M=T$, a tolerance $\mathrm{tol}>0$, and a relaxation parameter $\omega\in(0,1]$. Starting from $n_k^{S,(0)}(t_i)\equiv 1$, repeat the following steps until convergence.

\textbf{Forward step.} Given $n^{S,(m)}$, solve the forward Kolmogorov system for $$(S_k^{(m+1)},E_k^{(m+1)},I_k^{(m+1)},R_k^{(m+1)})$$ and compute $\Theta_k^{(m+1)}(t)$.

\textbf{Backward step.} Solve the HJB system backward from $t_M$ to $t_0$ to obtain $$(U_k^{S,(m+1)},U_k^{E,(m+1)},U_k^{I,(m+1)},U_k^{R,(m+1)}).$$

\textbf{Policy update.} Compute the closed-form minimizer $\tilde{n}_k^{S,(m+1)}$ via \eqref{eq:n_closed_form}, and apply relaxation:
\[
n_k^{S,(m+1)} \leftarrow (1-\omega)\,n_k^{S,(m)}+\omega\,\tilde n_k^{S,(m+1)}.
\]
Stop when $\max_{k,i}|n_k^{S,(m+1)}(t_i)-n_k^{S,(m)}(t_i)|<\mathrm{tol}$.

\subsection{Simulation Results}

We present results for homogeneous networks with parameters $\gamma=1$, $\sigma=2$, $r_I=50$, $\mathfrak{n}_{\text{min}}=0.1$, $\varepsilon=1$, and $\beta=4/k$ (so that $\beta k = 4$ is constant across degree classes, matching the normalization convention of \cite{Bremaud2025}). All simulations reach Nash equilibrium convergence in 35--44 FBS iterations (tolerance $10^{-7}$). We use a time step $\Delta t=0.05$. Figure~\ref{fig:seir_vs_sir} compares the SEIR-MFG and SIR-MFG Nash equilibria across degrees $k\in\{4,6,8,12,20\}$. The SEIR epidemic peaks later and the optimal effort $n^{S*}(t)$ stays closer to~1 (weaker precautions), consistent with Theorem~6.3. Figure~\ref{fig:sigma} shows the sensitivity of the equilibrium to the incubation rate $\sigma$: as $\sigma$ decreases (longer latency), the epidemic peak is further delayed and precautionary effort is progressively attenuated, illustrating the strategic delay mechanism.
\begin{figure}[htbp]
    \centering
    \includegraphics[width=\textwidth]{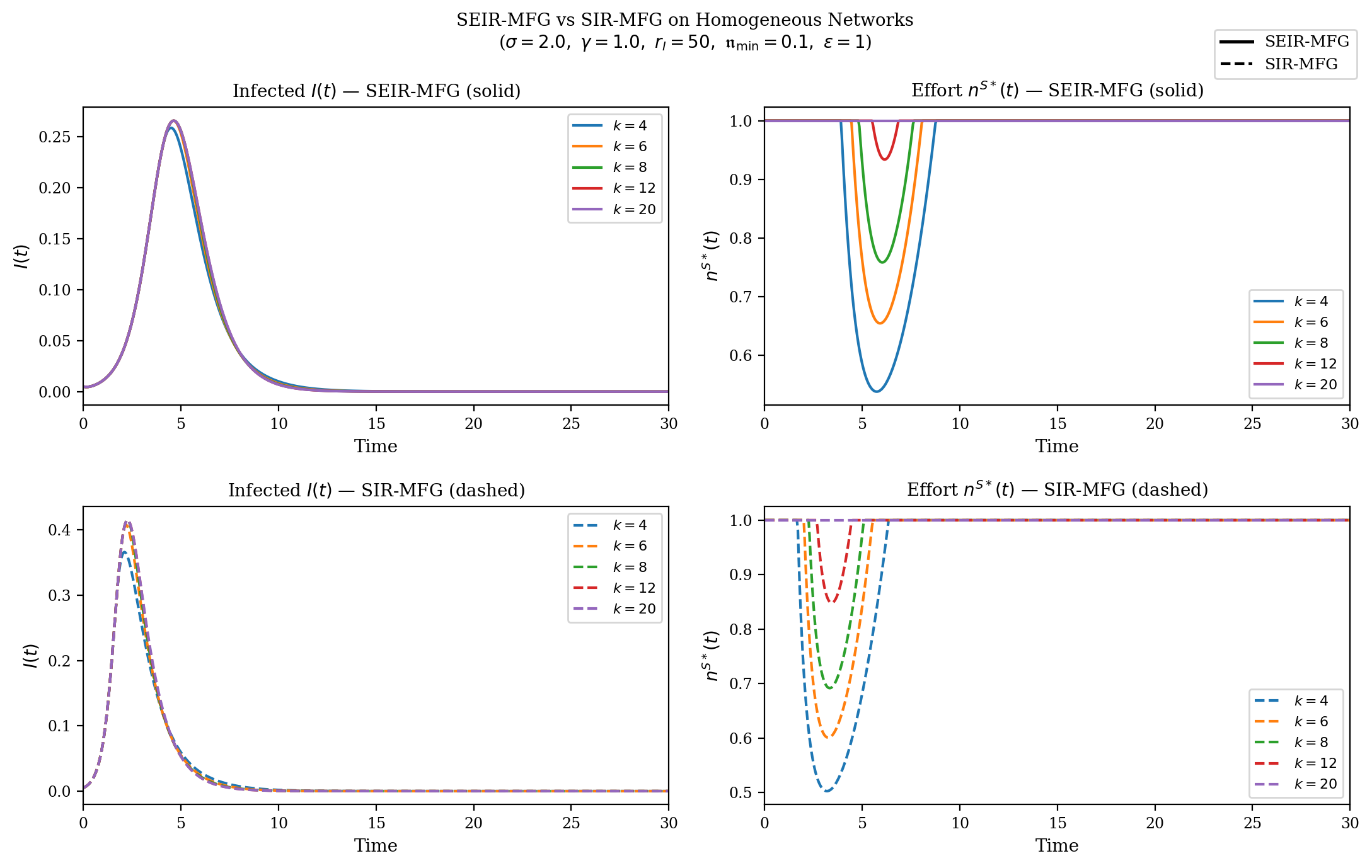}
    \caption{Nash equilibrium dynamics for SEIR-MFG (solid) and SIR-MFG (dashed) on homogeneous networks with $k\in\{4,6,8,12,20\}$. \textbf{Top row:} infected fraction $I(t)$ (left) and effort $n^{S*}(t)$ (right) for SEIR-MFG. \textbf{Bottom row:} same quantities for SIR-MFG. SEIR epidemics peak later and exhibit weaker behavioral responses, consistent with Theorem \ref{thm:comparison} (Early-phase growth and Attenuated Response).}
    \label{fig:seir_vs_sir}
\end{figure}

\begin{figure}[htbp]
    \centering
    \includegraphics[width=0.95\textwidth]{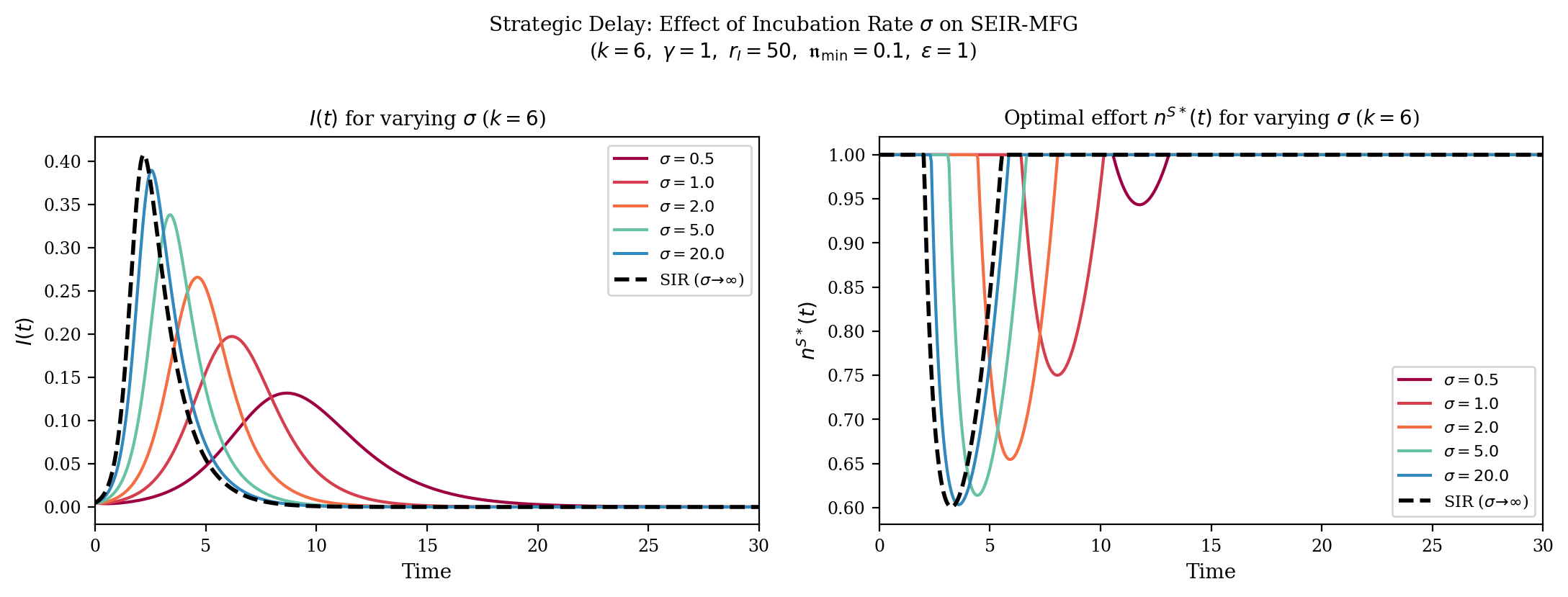}
    \caption{Effect of the incubation rate $\sigma$ on the SEIR-MFG Nash equilibrium ($k=6$). Decreasing $\sigma$ (longer latency) delays the epidemic peak and attenuates the optimal effort, quantifying the ``strategic delay'' effect. The SIR-MFG limit ($\sigma\to\infty$, dashed black) exhibits the strongest and earliest behavioral response.}
    \label{fig:sigma}
\end{figure}
Figure~\ref{fig:detailed} provides a detailed comparison for $k=8$, showing all four SEIR compartments, a direct overlay of SEIR-MFG vs.\ SIR-MFG effort, and the cumulative infection probability $\varphi(t)=1-S(t)$.
\begin{figure}[htbp]
    \centering
    \includegraphics[width=\textwidth]{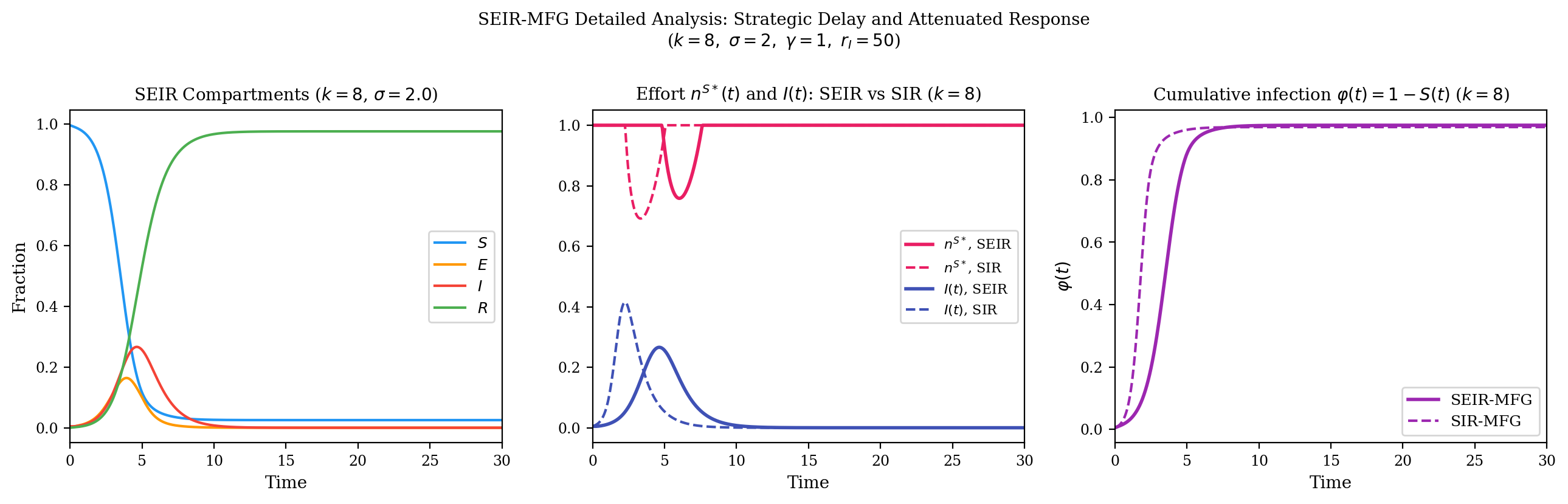}
    \caption{Detailed SEIR-MFG analysis for $k=8$, $\sigma=2$. \textbf{Left:} All four compartments $S,E,I,R$. \textbf{Center:} Optimal effort $n^{S*}(t)$ and infected fraction $I(t)$ for SEIR-MFG (solid) vs.\ SIR-MFG (dashed). \textbf{Right:} Cumulative infection probability $\varphi(t)=1-S(t)$; the larger final value for SEIR-MFG confirms Theorem~6.3 (Larger Final Size).}
    \label{fig:detailed}
\end{figure}
Figure~\ref{fig:bars} summarizes the peak timing and final outbreak size across all degree classes.
\begin{figure}[htbp]
    \centering
    \includegraphics[width=0.95\textwidth]{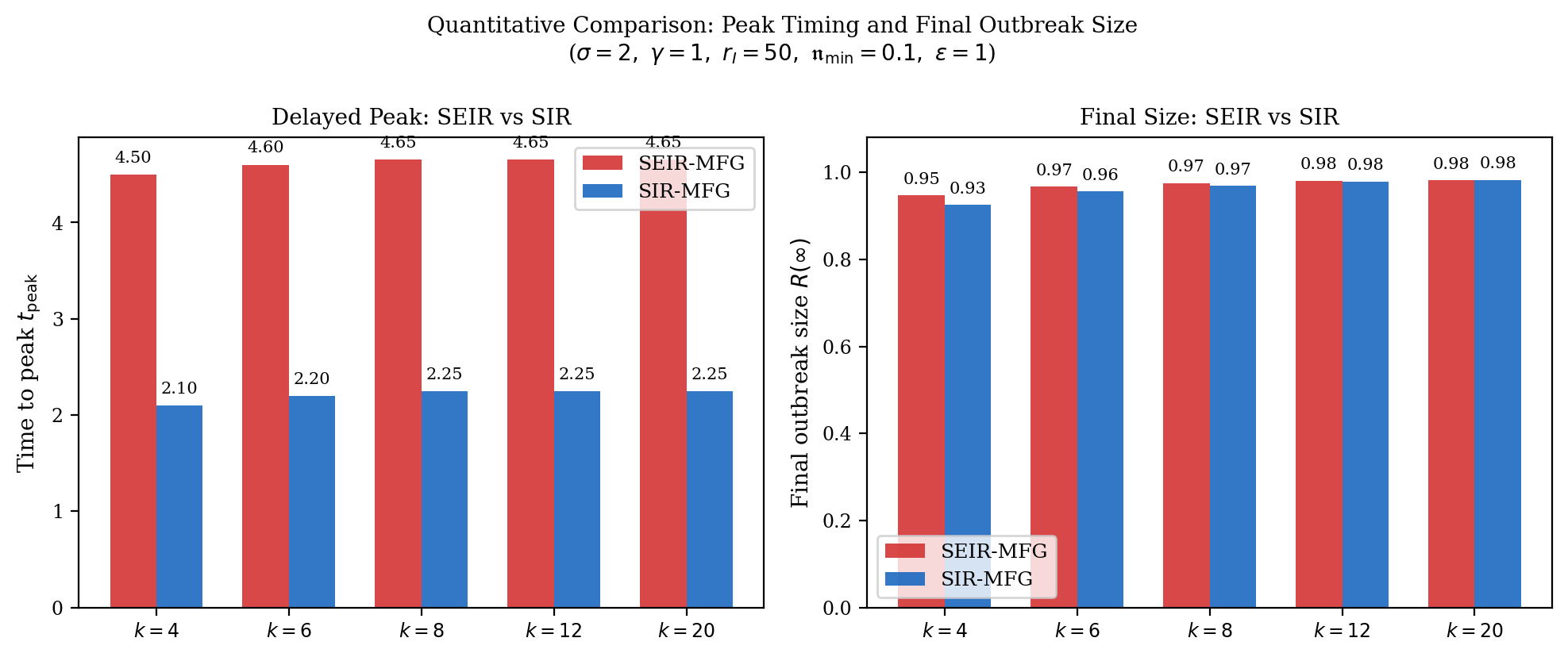}
    \caption{Quantitative comparison of peak timing $t_{\mathrm{peak}}$ (left) and final outbreak size $R(\infty)$ (right) for SEIR-MFG vs.\ SIR-MFG across all degree classes $k\in\{4,6,8,12,20\}$. SEIR-MFG peaks later and yields a larger final size for $k\leq 12$. For $k=20$, both models saturate at $n^{S*}=\mathfrak{n}_{\min}$, removing strategic differentiation; see Table in Section~\ref{sec:numerical}.}
    \label{fig:bars}
\end{figure}
The numerical results confirm the claims of Theorem~\ref{thm:comparison}. The quantitative summary (simulation parameters: $\Delta t=0.05$, $\omega=0.3$, tol $=10^{-7}$; converging in 35--44 FBS iterations for $k\le 12$) is:

\begin{center}
\begin{tabular}{r|cc|cc}
$k$ & $t_{\mathrm{peak}}^{\mathrm{SEIR}}$ & $t_{\mathrm{peak}}^{\mathrm{SIR}}$ & $R^{\mathrm{SEIR}}(\infty)$ & $R^{\mathrm{SIR}}(\infty)$ \\
\hline
4  & 4.50 & 2.10 & 0.9465 & 0.9253 \\
6  & 4.60 & 2.20 & 0.9667 & 0.9558 \\
8  & 4.65 & 2.25 & 0.9749 & 0.9691 \\
12  & 4.65 & 2.25 & 0.9810 & 0.9792 \\
20  & 4.65 & 2.25 & 0.9817 & 0.9825 \\
\end{tabular}
\end{center}
{\footnotesize $^{\dagger}$For $k=20$ with $\varepsilon=1$, the infection pressure $\beta k\Theta\Delta U > 1$, so both SEIR and SIR models saturate at the lower boundary $n^{S*}=\mathfrak{n}_{\min}=0.1$ for most of the epidemic. Dynamics become essentially identical (converged in 1 iteration), and $R^{\mathrm{SIR}}(\infty)$ slightly exceeds $R^{\mathrm{SEIR}}(\infty)$ because the boundary saturation removes the strategic differentiation.}

\section{Conclusion and Future Directions}
\label{sec:conclusion}

In this paper, we developed a mean-field game formulation for SEIR dynamics on heterogeneous contact networks with degree-dependent behavior. We derived the coupled forward--backward equilibrium system and obtained explicit best-response structure for susceptible contact effort under the convex isolation cost. A key qualitative consequence of the exposed compartment is that the incubation stage separates infection from infectiousness and shifts incentives over time; in the baseline formulation, exposed agents optimally maintain full contact, while susceptible agents reduce contact according to a value-gap criterion. We proved existence of equilibrium via a Schauder fixed-point argument and provided a uniqueness result under an appropriate monotonicity condition. Numerical experiments illustrate how degree heterogeneity, incubation, and the cost exponent influence equilibrium effort profiles and epidemic outcomes.

Several extensions are natural and would broaden the applicability of the framework. First, one can enrich the information structure: agents may observe only noisy aggregate signals, which leads to partially observed control and filtering-coupled HJB systems. Second, one can incorporate additional policy levers---testing, isolation mandates, and contact tracing---that act directly on exposed individuals and create state-dependent constraints or penalties; this would clarify when and how nontrivial exposed effort can arise as an equilibrium outcome. Third, the network environment can be made time-dependent by allowing mobility-driven mixing or adaptive rewiring, so that contact opportunities evolve endogenously with perceived risk and public interventions. Fourth, heterogeneity beyond degree can be incorporated to study targeted incentives and distributional outcomes.  Finally, empirical calibration and validation using behavioral and epidemic data would help identify which isolation-cost specifications reproduce observed contact reductions across connectivity classes.

Lastly, another future avenue will be incorporating more types of human behavior, such as compliance vs.\ noncompliance and age-structured population. The SIR-based models incorporating such behavior have been studied from various aspects including ODE, reaction-diffusion PDEs, optimal control, and stochastic differential equations \cite{PW1, PW2, PW3, PW4, pang1, pang2, pang3}.

\section*{Acknowledgments}
This research was partially supported by the Simons Foundation Travel Support for Mathematicians (No. 0007730).

	\bibliographystyle{alpha}

\begin{thebibliography}{99}
		

		
		\bibitem{Bremaud2025}
L.~Bremaud, O.~Giraud, and D.~Ullmo.
\newblock Mean-field game approach to epidemic propagation on networks.
\newblock {\em Phys. Rev. E}, 112(5):L052301, 2025.
\newblock doi:10.1103/mys6-fznc.
		
		\bibitem{Liu2023}
		H.~Liu and X.~Tian.
		\newblock Data-driven optimal control of a SEIR model for COVID-19.
		\newblock {\em Communications on Pure and Applied Analysis}, 22(1):19--39, 2023.
		
		\bibitem{Sherborne2018}
		N.~Sherborne, J.~C. Miller, K.~B. Blyuss, and I.~Z. Kiss.
		\newblock Mean-field models for non-Markovian epidemics on networks.
		\newblock {\em Journal of Mathematical Biology}, 76:755--778, 2018.
		
		\bibitem{Lasry2007}
		J.-M. Lasry and P.-L. Lions.
		\newblock Mean field games.
		\newblock {\em Japanese Journal of Mathematics}, 2(1):229-260, 2007.
		
		\bibitem{Elie2020}
		R.~Elie, E.~Hubert, and G.~Turinici.
		\newblock Contact rate epidemic control of COVID-19: an equilibrium view.
		\newblock {\em Mathematical Modelling of Natural Phenomena}, 15:35, 2020.
		
		\bibitem{Chang2020}
		S.~L. Chang, M.~Piraveenan, P.~Pattison, and M.~Prokopenko.
		\newblock Game theoretic modelling of infectious disease dynamics and intervention methods: a review.
		\newblock {\em Journal of Biological Dynamics}, 14(1):57-89, 2020.
		
		\bibitem{Li2025}
		J.~Li, Z.~Jin, and M.~Tang.
		\newblock Analysis of the SEIR mean-field model in dynamic networks under intervention.
		\newblock {\em Infectious Disease Modelling}, 10:850--874, 2025.

        \bibitem{PW3}
C.~Parkinson and W.~Wang,  A compartmental model for epidemiology with human behavior and stochastic effects, \textit{Math. Biosci.}, vol.~392, p.~109588, 2026.

\bibitem{PW1}
C.~Parkinson and W.~Wang, Analysis of a reaction-diffusion SIR epidemic model with noncompliant behavior, \textit{SIAM J. Appl. Math.}, vol.~83, no.~5, pp.~1969--2002, 2023.

\bibitem{PW2}
M.~Bongarti, C.~Parkinson, and W.~Wang, Optimal control of a reaction-diffusion epidemic model with non-compliance, \textit{Eur. J. Appl. Math.}, 2025.

\bibitem{ParkinsonRoy2026FP}
C.~Parkinson and S.~Roy.
\newblock \emph{A Fokker--Planck Framework for Control of Epidemics}.
\newblock arXiv preprint arXiv:2601.20181, 2026.

\bibitem{pang1}
G.~Pang and \'{E}.~Pardoux, {PDE} model for multi-patch epidemic models with migration and infection-age dependent infectivity, \emph{Pure Appl. Funct. Anal.} \textbf{9} (2024), no.~3, 863--897.

\bibitem{pang2}
G.~Pang and \'{E}.~Pardoux, {SPDE} for stochastic {SIR} epidemic models with infection-age dependent infectivity, \emph{Stochastics and Partial Differential Equations: Analysis and Computations} (2025), 1--54, {doi:10.1007/s40072-025-00403-x}.

\bibitem{pang3}
G.~Pang, \'{E}.~Pardoux, and A.~Velleret, Stochastic {SIR} model with individual heterogeneity and infection-age dependent infectivity on large non-homogeneous random graphs, \emph{arXiv preprint} arXiv:2502.04225 (2025).

\bibitem{PW4}
C.~Ngo, C.~Parkinson, and W.~Wang, Optimal control of an SIR model with noncompliance as a social contagion, \textit{arXiv preprint arXiv:2509.09075}, 2025.

\bibitem{Kermack1927}
W.~O. Kermack and A.~G. McKendrick.
\newblock A contribution to the mathematical theory of epidemics.
\newblock {\em Proc. Roy. Soc. London Ser. A}, 115:700--721, 1927.
\newblock doi:10.1098/rspa.1927.0118.

\bibitem{AndersonMay1991}
R.~M. Anderson and R.~M. May.
\newblock {\em Infectious Diseases of Humans: Dynamics and Control}.
\newblock Oxford University Press, 1991.

\bibitem{Funk2010}
S.~Funk, M.~Salath\'e, and V.~A.~A. Jansen.
\newblock Modelling the influence of human behaviour on the spread of infectious diseases: a review.
\newblock {\em J. R. Soc. Interface}, 7(50):1247--1256, 2010.
\newblock doi:10.1098/rsif.2010.0142.

\bibitem{Verelst2016}
F.~Verelst, L.~Willem, and P.~Beutels.
\newblock Behavioural change models for infectious disease transmission: a systematic review (2010--2015).
\newblock {\em J. R. Soc. Interface}, 13(125):20160820, 2016.
\newblock doi:10.1098/rsif.2016.0820.

\bibitem{Fenichel2011}
E.~P. Fenichel, C.~Castillo-Chavez, M.~G. Ceddia, et al.
\newblock Adaptive human behavior in epidemiological models.
\newblock {\em Proc. Natl. Acad. Sci. USA}, 108(15):6306--6311, 2011.
\newblock doi:10.1073/pnas.1011250108.

\bibitem{Reluga2010}
T.~C. Reluga.
\newblock Game theory of social distancing in response to an epidemic.
\newblock {\em PLoS Comput. Biol.}, 6(5):e1000793, 2010.
\newblock doi:10.1371/journal.pcbi.1000793.

\bibitem{Bauch2004}
C.~T. Bauch and D.~J.~D. Earn.
\newblock Vaccination and the theory of games.
\newblock {\em Proc. Natl. Acad. Sci. USA}, 101(36):13391--13394, 2004.
\newblock doi:10.1073/pnas.0403823101.

\bibitem{Eksin2017}
C.~Eksin, J.~S. Shamma, and J.~S. Weitz.
\newblock Disease dynamics in a stochastic network game: a little empathy goes a long way in averting outbreaks.
\newblock {\em Sci. Rep.}, 7:44122, 2017.
\newblock doi:10.1038/srep44122.

\bibitem{KeelingEames2005}
M.~J. Keeling and K.~T.~D. Eames.
\newblock Networks and epidemic models.
\newblock {\em J. R. Soc. Interface}, 2(4):295--307, 2005.
\newblock doi:10.1098/rsif.2005.0051.

\bibitem{Newman2002Spread}
M.~E.~J. Newman.
\newblock Spread of epidemic disease on networks.
\newblock {\em Phys. Rev. E}, 66:016128, 2002.
\newblock doi:10.1103/PhysRevE.66.016128.

\bibitem{PastorSatorras2015}
R.~Pastor-Satorras, C.~Castellano, P.~Van Mieghem, and A.~Vespignani.
\newblock Epidemic processes in complex networks.
\newblock {\em Rev. Mod. Phys.}, 87:925--979, 2015.
\newblock doi:10.1103/RevModPhys.87.925.

\bibitem{Volz2008}
E.~Volz.
\newblock SIR dynamics in random networks with heterogeneous connectivity.
\newblock {\em J. Math. Biol.}, 56(3):293--310, 2008.
\newblock doi:10.1007/s00285-007-0116-4.

\bibitem{Boguna2002}
M.~Bogu\~n\'a and R.~Pastor-Satorras.
\newblock Epidemic spreading in correlated complex networks.
\newblock {\em Phys. Rev. E}, 66:047104, 2002.
\newblock doi:10.1103/PhysRevE.66.047104.

\bibitem{Kiss2017}
I.~Z. Kiss, J.~C. Miller, and P.~L. Simon.
\newblock {\em Mathematics of Epidemics on Networks: From Exact to Approximate Models}.
\newblock Interdisciplinary Applied Mathematics, vol.~46. Springer, 2017.
\newblock doi:10.1007/978-3-319-50806-1.

\bibitem{Huang2006}
M.~Huang, R.~P. Malham\'e, and P.~E. Caines.
\newblock Large population stochastic dynamic games: closed-loop McKean--Vlasov systems and the Nash certainty equivalence principle.
\newblock {\em Commun. Inf. Syst.}, 6(3):221--252, 2006.

\bibitem{CarmonaDelarue2018}
R.~Carmona and F.~Delarue.
\newblock {\em Probabilistic Theory of Mean Field Games with Applications I--II}.
\newblock Probability Theory and Stochastic Modelling, vols.~83--84. Springer, 2018.

\bibitem{GomesSaude2014}
D.~A. Gomes and J.~Sa\'ude.
\newblock Mean field games models---a brief survey.
\newblock {\em Dyn. Games Appl.}, 4(2):110--154, 2014.
\newblock doi:10.1007/s13235-013-0099-2.

\bibitem{Achdou2010}
Y.~Achdou and I.~Capuzzo-Dolcetta.
\newblock Mean field games: numerical methods.
\newblock {\em SIAM J. Numer. Anal.}, 48(3):1136--1162, 2010.
\newblock doi:10.1137/090758477.

\bibitem{Cardaliaguet2017}
P.~Cardaliaguet and S.~Hadikhanloo.
\newblock Learning in mean field games: the fictitious play.
\newblock {\em ESAIM Control Optim. Calc. Var.}, 23(2):569--591, 2017.
\newblock doi:10.1051/cocv/2016004.

\bibitem{CarmonaLauri2019}
R.~Carmona, M.~Lauri\`ere, and Z.~Tan.
\newblock Model-free mean-field reinforcement learning: mean-field MDP and mean-field Q-learning.
\newblock {\em arXiv:1910.12802}, 2019.

\bibitem{Elie2020}
R.~Elie, E.~Hubert, and G.~Turinici.
\newblock Contact rate epidemic control of COVID-19: an equilibrium view.
\newblock {\em Math. Model. Nat. Phenom.}, 15:35, 2020.
\newblock doi:10.1051/mmnp/2020022.

\bibitem{Cho2020}
S.~Cho.
\newblock Mean-field game analysis of SIR model with social distancing.
\newblock {\em arXiv:2005.06758}, 2020.

\bibitem{Bremaud1981}
P.~Br\'emaud.
\newblock {\em Point Processes and Queues: Martingale Dynamics}.
\newblock Springer, 1981.

\bibitem{Bremaud2022}
L.~Bremaud and D.~Ullmo.
\newblock Social structure description of epidemic propagation with a mean-field game paradigm.
\newblock {\em Phys. Rev. E}, 106:L062301, 2022.
\newblock doi:10.1103/PhysRevE.106.L062301.

\bibitem{Bremaud2024NPI}
L.~Bremaud, O.~Giraud, and D.~Ullmo.
\newblock Mean-field-game approach to nonpharmaceutical interventions in a social-structure model of epidemics.
\newblock {\em Phys. Rev. E}, 110:064301, 2024.
\newblock doi:10.1103/PhysRevE.110.064301.


\bibitem{He2020}
X.~He, E.~H. Lau, P.~Wu, et al.
\newblock Temporal dynamics in viral shedding and transmissibility of COVID-19.
\newblock {\em Nat. Med.}, 26(5):672--675, 2020.
\newblock doi:10.1038/s41591-020-0869-5.
		
	\end{thebibliography}

\end{document}